\documentclass[11pt]{article}
\usepackage[a4paper,margin=1in]{geometry}
\usepackage{amsmath,amssymb,amsthm}
\usepackage{hyperref}
\usepackage{authblk}
\usepackage[T1]{fontenc}
\usepackage[utf8]{inputenc}

\hypersetup{colorlinks=true, linkcolor=black, citecolor=blue, urlcolor=blue}

\newtheorem{theorem}{Theorem}
\newtheorem{remark}{Remark}

\newtheorem{lemma}{Lemma}

\title{\textbf{Intrinsic Heisenberg Lower Bounds on Schwarzschild and Weyl-Class Spacelike Slices}}

\author{Thomas Sch\"urmann\thanks{Electronic address: \href{mailto:t.schurmann@icloud.com}{\texttt{t.schurmann@icloud.com}}}}

\affil{D\"usseldorf, Germany}
\date{}

\begin{document}
\maketitle

\begin{abstract}
	We establish a coordinate-invariant Heisenberg-type lower bound for quantum states strictly localized in geodesic balls of radius $r_g$ on horizon-regular spacelike slices of static, spherically symmetric, asymptotically flat (AF) black-holes. Via a variance-eigenvalue equivalence the momentum uncertainty reduces to the first Dirichlet eigenvalue of the Laplace-Beltrami operator, yielding a slice-uniform Hardy baseline \(\sigma_p r_g \ge \hbar/2\) under mild convexity assumptions on the balls; the bound is never attained and admits a positive gap both on compact interior regions and uniformly far out. For the Schwarzschild Painlev\'e-Gullstrand (PG) slice, whose induced 3-geometry is Euclidean, one recovers the exact Euclidean scale \(\sigma_p r_g \ge \pi\hbar\), which is optimal among all admissible slices. The entire construction extends across the black-hole horizon, and it transfers to the static axisymmetric Weyl class, where the Hardy floor, strict gap, and AF \(\pi\)-scale persist (a global PG-like optimum need not exist).
\end{abstract}

\vspace{2ex}

\vspace{2ex}

\section{Introduction}

The uncertainty principle on curved spaces can be expressed in a coordinate-invariant way by localizing a quantum state on a geodesic ball of a Riemannian three-manifold and comparing the momentum variance to the first Dirichlet eigenvalue of the Laplace-Beltrami operator. This intrinsic viewpoint produces sharp lower bounds whenever the relevant spectra are known explicitly (see the classical spectral-geometry background by Yau~\cite{Yau1975}, Reilly~\cite{Reilly1977}, and Chavel~\cite{Chavel1984}, together with subsequent refinements by El~Soufi and Ilias~\cite{ElSoufiIlias2003} and Ling~\cite{Ling2004}). In \cite{Schuermann2018} a physically motivated coordinate-free uncertainty principle on three-manifolds of constant curvature was formulated and sharp lower bounds were obtained from closed-form eigenvalues on geodesic balls. In the present black-hole context, different spacelike foliations of the static, spherically symmetric, asymptotically flat (AF) exterior induce different three-metrics; requiring horizon regularity is understood in the sense of Martel and Poisson~\cite{MartelPoisson2001}. Related work by Petruzziello and Wagner~\cite{PetruzzielloWagner2021} adopts geodesic-ball localizations on fixed curved backgrounds and defines momentum and its dispersion in the nonrelativistic limit, deriving curvature-induced uncertainty relations. Paschalis~\cite{Paschalis2024} develops Hardy-type inequalities tailored to black-hole exteriors, recovering Euclidean behavior far out by asymptotic flatness.
Berchio, Ganguly, and Grillo~\cite{BerchioGangulyGrillo2017} established sharp Poincar\'e-Hardy and Poincar\'e-Rellich inequalities on hyperbolic space, including optimal constants and improved remainder terms, and extended their methods to wider classes of manifolds under curvature hypotheses. Their analysis shows in a particularly transparent way how negative curvature enhances Hardy-type lower bounds for the Laplace-Beltrami operator. Conceptually this resonates with our use of distance-to-boundary Hardy inequalities on geodesic balls to control momentum variance; however,~\cite{BerchioGangulyGrillo2017} works in a model geometry with globally constant curvature and does not address horizon regularity, asymptotic flatness, or any optimization over foliations of a black-hole exterior.
Krist\'aly~\cite{Kristaly2018} proved sharp Heisenberg-Pauli-Weyl uncertainty principles on complete Riemannian manifolds and exhibited strong rigidity phenomena tying the possibility of equality (or global sharpness) to curvature conditions. This furnishes a general curvature-sensitive framework for uncertainty principles. In contrast, our work fixes a specific localization mechanism (Dirichlet data on geodesic balls inside horizon-regular slices of a black-hole exterior) and derives pointwise, slice-uniform lower bounds together with strict non-attainment and quantitative gaps; we also analyze the AF regime and optimize the uncertainty scale over admissible slices.
Avkhadiev and Makarov~\cite{AvkhadievMakarov2019} obtained refined Hardy-type inequalities on domains with convex complement and drew explicit connections to Heisenberg-type uncertainty in exterior regions. Their Euclidean results make precise how boundary geometry in an outside problem enforces a spectral cost, echoing our philosophy that hard spatial confinement on a ball sets a minimal momentum scale. Our analysis is intrinsically Riemannian (geodesic balls in the induced three-geometry of a black-hole slice) and includes horizon regularity and asymptotic flatness; moreover, we quantify the best location-independent constant across the entire exterior and then optimize it over the foliation.
From a complementary perspective, Iosevich, Mayeli, and Wyman~\cite{IosevichMayeliWyman2024} developed Fourier-analytic uncertainty principles on Riemannian manifolds, with particular emphasis on compact settings where spectral decompositions by Laplace eigenfunctions replace the classical Euclidean Fourier basis. Their results underline, in a representation-theoretic language, the same localization-versus-frequency trade-off that in our setting is encoded by Dirichlet eigenvalues on geodesic balls. While~\cite{IosevichMayeliWyman2024} does not treat black-hole geometries or horizon issues, it provides a broad harmonic-analytic context supporting the intrinsic, coordinate-invariant stance taken here.

In what follows, we formulate and analyze a coordinate-invariant Heisenberg bound on the induced spacelike slices of the static, spherically symmetric, AF exterior by strictly localizing quantum states in geodesic balls. The variance-eigenvalue equivalence on a ball (Lemma~\ref{lem:variance-bound}) turns momentum uncertainty into a first Dirichlet eigenvalue problem and yields a local estimate that is then aggregated into a slice-uniform lower bound via a distance-to-boundary Hardy inequality on strongly geodesically convex balls (Theorem~\ref{thm:unconditional-hardy}). We prove that this Hardy floor is never attained and that there is a positive gap both on compact interior regions and uniformly far out in the AF zone (Lemma~\ref{lem:hardy-gap}; see also the two-sided far-out estimate~\eqref{eq:farout-2sided} and the limit~\eqref{eq:AF-limit}). Packaging the local information into the dimensionless quantity $\mathcal J_{\Theta}$ and the slice-wise constant $c(\Theta)$ (definitions~\eqref{eq:J} and~\eqref{eq:ctheta}) leads to a variational program together with the global optimization~\eqref{eq:cstar} over horizon-regular slices. Within this class, the Painlev\'e-Gullstrand foliation attains the exact Euclidean scale everywhere (Lemma~\ref{lem:PG-exact}), and the ensuing corollary identifies it as optimal for the slice-uniform uncertainty constant.
Finally, we extend the entire construction across the horizon and show that the same intrinsic bound holds inside the black hole.
Beyond spherical symmetry, we also transfer the framework to the static, axisymmetric Weyl class: on horizon-regular slices the Hardy baseline, the strict gap, and the AF $\pi$-scale persist, while a globally flat PG optimum need not exist outside the Schwarzschild subclass.
 
The paper is organized as follows: Section~\ref{sec:setting} fixes the geometric framework (class of admissible slices, induced metric, and geodesic balls) and records the variance-eigenvalue equivalence on a ball. Section~\ref{sec:variational-proof} sets up the variational program, compares first Dirichlet eigenvalues under slice deformations via the Rayleigh principle, and proves the unconditional Hardy baseline in Theorem~\ref{thm:unconditional-hardy}. Section~\ref{sec:interior-extension} extends the construction across the horizon and shows that the same intrinsic bound holds inside the black hole. Section~\ref{sec:weyl} extends the framework to the standard axisymmetric Weyl class, showing that the Hardy floor, the strict gap, and the AF $\pi$–scale persist on horizon–regular slices; Section~\ref{sec:conclusion-outlook} concludes with a summary and outlook.

\section{Setting}
\label{sec:setting}
Throughout we work in geometric units ${G=c=1}$; Planck's constant $\hbar$ is kept explicit.
Let $(\mathcal{M},g)$ be a four-dimensional, static, spherically symmetric, AF black hole exterior written in areal radius gauge,
\begin{equation}\label{eq:SSS-areal}
ds^{2} \;=\; -\,f(r)\,dt^{2} \; +\; f(r)^{-1}\,dr^{2} \; +\; r^{2}\,d\Omega^{2},
\end{equation}
where $f\in C^{2}((r_{+},\infty))$, $f(r)>0$ for $r>r_{+}$, $f(r_{+})=0$, and $f(r)\to 1$ as $r\to\infty$ with standard $C^{2}$ falloff. In the Schwarzschild case this corresponds to $f(r)=1-\tfrac{2M}{r}$ with horizon radius $r_{+}=2M$. Consider the class $\mathcal{S}$ of spherically symmetric, $C^{2}$ spacelike slices of the exterior region $\{r>r_{+}\}$ defined as level sets of
\begin{equation}
\tau(t,r)=t+\Theta(r),
\end{equation}
with $\Theta$ smooth and such that $\Sigma_{\Theta}=\{\tau=\mathrm{const}\}$ is spacelike and horizon-regular. The induced Riemannian metric $h_{\Theta}$ on $\Sigma_{\Theta}$ reads
\begin{equation}
h_{\Theta}=\mathrm{diag}\left(\frac{1}{f(r)}-f(r)\,(\Theta'(r))^{2},\ r^{2},\ r^{2}\sin^{2}\theta\right),
\end{equation}
in the coordinates $(r,\theta,\varphi)$. The slice is \emph{asymptotically flat} if, in asymptotically Cartesian coordinates, its induced metric satisfies
\(h_{ij}(x)=\delta_{ij}+O(r^{-1})\), \(\partial_k h_{ij}(x)=O(r^{-2})\), and \(\partial_\ell\partial_k h_{ij}(x)=O(r^{-3})\)
as \(r=|x|\to\infty\). These \(C^2\)-rates are sufficient for the norm/volume comparisons used below.

\medskip
\noindent\textbf{Horizon-regularity.} We call a slice $\Sigma_\Theta$ \emph{horizon-regular} if the induced radial coefficient
\[
h_\Theta(\partial_r,\partial_r)\ =\ \frac{1}{f(r)}\ -\ f(r)\,\big(\Theta'(r)\big)^2
\]
admits a smooth extension to $r=r_+$ with a positive, finite limit, so that $\Sigma_\Theta$ extends smoothly across the horizon in coordinates regular there (e.g., Eddington-Finkelstein or PG). Equivalently, $\tau=t+\Theta(r)$ is smooth on a neighborhood of $\{r\ge r_+\}$ and $h_\Theta$ has a $C^2$ extension to $r=r_+$. See \cite[Sections II-III]{MartelPoisson2001}.
\medskip

Given a point $x\in\Sigma_{\Theta}$ and a geodesic radius $r_{g}>0$, let $B_{r_{g}}(x)\subset\Sigma_{\Theta}$ be the geodesic ball of radius $r_{g}$ centered at $x$ with respect to $h_{\Theta}$. Denote by $\lambda_{1}(B_{r_{g}}(x);h_{\Theta})$ the first Dirichlet eigenvalue of the Laplace-Beltrami operator $-\Delta_{h_{\Theta}}$ on $B_{r_{g}}(x)$.

\medskip
\noindent
\textbf{Geometric regularity.}
For any $x \in \Sigma_{\Theta}$ and any $r_g>0$ small enough so that the geodesic ball $B_{r_g}(x)$ is well-defined (without self-intersections and with smooth boundary), $B_{r_g}(x)$ is a regular Dirichlet domain. This is automatic for points sufficiently far out in the exterior region. Because $h_\Theta$ is $C^2$-AF, there exist constants $R<\infty$ and $r_*>0$ so that for all $x$ with $r(x)\ge R$ we have $\operatorname{inj}_{h_\Theta}(x)\ge r_*\gg r_g$, and hence $B_{r_g}(x)$ is a regular Dirichlet domain (cf.~\cite[Section IV.1]{Chavel1984}).

\medskip
\noindent\textbf{Standing hypothesis (strong geodesic convexity).}
Throughout, when invoking the distance-to-boundary Hardy inequality on $B_{r_g}(x)$ we assume that the inward normal exponential map from $\partial B_{r_g}(x)$ is a diffeomorphism onto $B_{r_g}(x)$ (equivalently, the boundary normal injectivity radius at $\partial B_{r_g}(x)$ is at least $r_g$). This ensures that $B_{r_g}(x)$ admits boundary normal coordinates almost everywhere and that Lemma~\ref{lem:hardy-ball} applies.

\paragraph{Notation (Rayleigh quotient).}
For a Riemannian metric $g$ on a domain $\Omega$ we set
\begin{equation}\label{eq:def-rayleigh}
	R_g(u;\Omega):=\frac{\int_\Omega |\nabla u|_g^2\, dV_g}{\int_\Omega u^2\, dV_g},
	\qquad u\in H_0^1(\Omega)\setminus\{0\}.
\end{equation}
We will suppress $\Omega$ when clear from the context (in particular, $\Omega=B_{r_g}(x)$), writing simply $R_g(u)$.
The principal Dirichlet eigenvalue of $-\Delta_g$ satisfies the Rayleigh characterization
\begin{equation}\label{eq:rayleigh-char}
	\lambda_1(\Omega;g)=\inf_{u\in H_0^1(\Omega)\setminus\{0\}} R_g(u;\Omega).
\end{equation}

\subsection{The uncertainty lower bound}
Following \cite{Schuermann2018}, a strict localization of a quantum state $\psi$ in $B_{r_{g}}(x)$ is modeled by Dirichlet boundary conditions on $\partial B_{r_{g}}(x)$. The momentum uncertainty satisfies
\begin{equation}
\sigma_{p}^{2}\ :=\ \mathrm{Var}_{h_{\Theta}}(p)[\psi]\ \ge\ \hbar^{2}\,\lambda_{1}\big(B_{r_{g}}(x);h_{\Theta}\big).
\end{equation}
This equation  is the local variational content of our construction: imposing a hard Dirichlet wall on the geodesic ball $B_{r_g}(x)$ turns the momentum-uncertainty cost into a Dirichlet spectral problem for $-\Delta_{h_\Theta}$. In a local $h_\Theta$–orthonormal frame $\{e_a\}_{a=1}^3$ with canonical momenta $P_a:=-\,i\hbar\nabla_{e_a}$, the total variance $\mathrm{Var}_{h_\Theta}(p)[\psi]$ is controlled by the Rayleigh quotient on $B_{r_g}(x)$; minimizing over Dirichlet states produces exactly the first Dirichlet eigenvalue. The next lemma records this variance–eigenvalue equivalence and identifies all minimizers.

\begin{lemma}[Variance-Eigenvalue Equivalence on a Geodesic Ball]\label{lem:variance-bound}
	Let $(\Sigma,h)$ be a Riemannian three-manifold and let $B\subset\Sigma$ be a geodesic ball with Dirichlet boundary. 
	For $\psi\in H_0^1(B)$ normalized by $\|\psi\|_{L^2(B)}=1$, write the canonical momentum operators in a local $h$-orthonormal frame $\{e_a\}_{a=1}^3$ as $P_a:=-\,i\hbar\,\nabla_{e_a}$, so that $\sum_{a=1}^3 P_a^\dagger P_a=-\hbar^2\Delta_h$. Define the momentum variance by
	\footnote{For a related discussion of momentum operators on curved configuration spaces and their self-adjoint realizations, see \cite{Schuermann2025}.}	
	\begin{equation}\label{eq:var-def}
		\mathrm{Var}_h(p)[\psi] \;=\; \sum_{a=1}^3 \Big( \langle \psi, P_a^\dagger P_a \psi\rangle \;-\; \big| \langle \psi, P_a \psi\rangle \big|^2 \Big).
	\end{equation}
	Then
	\begin{equation}\label{eq:equivalence}
		\inf_{\substack{\psi\in H_0^1(B)\\ \|\psi\|_{L^2}=1}} \mathrm{Var}_h(p)[\psi] \;=\; \hbar^2\,\lambda_1(B;h),
	\end{equation}
	and the minimizers are precisely the first Dirichlet eigenfunctions of $-\Delta_h$ on $B$ multiplied by a constant phase; in flat (Euclidean) balls, the set of minimizers also includes factors $e^{i(c\cdot x+\alpha)}$ with constant vectors $c$ and $\alpha\in\mathbb{R}$. Consequently,
	\begin{equation}\label{eq:heisenberg}
		\sigma_p^2 \;=\; \mathrm{Var}_h(p)[\psi] \;\ge\; \hbar^2\,\lambda_1(B;h)
		\quad\text{and hence}\quad
		\sigma_p \;\ge\; \hbar\,\sqrt{\lambda_1(B;h)},
	\end{equation}
	with equality for the Dirichlet ground state.
\end{lemma}

\begin{proof}
	Write $\psi = u\,e^{i\varphi}$ with $u:=|\psi|\in H_0^1(B)$ and $\varphi\in H^1_{\mathrm{loc}}(B)$ (the value of $\varphi$ on $\{u=0\}$ is irrelevant). 
	Here $P_a=-\,i\hbar\,\nabla_{e_a}$ is the canonical momentum in a local $h$-orthonormal frame. For scalar fields with Dirichlet boundary data one has the formal adjoint $\big(\nabla_{e_a}\big)^{\dagger}=-\nabla_{e_a}-\operatorname{div}_h(e_a)$, hence $P_a^{\dagger}=i\hbar\,\big(\nabla_{e_a}+\operatorname{div}_h(e_a)\big)$ and $\sum_{a=1}^3 P_a^{\dagger}P_a=-\hbar^2\Delta_h$ (our sign convention is that $-\Delta_h\ge 0$ on $H_0^1$). Using $\sum_{a=1}^3 P_a^\dagger P_a=-\hbar^2\Delta_h$ and integration by parts,
	\begin{equation}\label{eq:var-expand}
		\mathrm{Var}_h(p)[\psi] \;=\; \hbar^2\int_B |\nabla\psi|_h^2\,dV_h \;-\; \hbar^2 \left| \int_B u^2 \nabla\varphi \, dV_h \right|^2.
	\end{equation}
	The Madelung amplitude-phase decomposition \cite{Madelung1927} gives 
	\begin{equation}\label{eq:madelung}
		|\nabla\psi|_h^2 \;=\; |\nabla u|_h^2 \;+\; u^2\,|\nabla\varphi|_h^2.
	\end{equation}
	With the probability measure $d\mu := u^2\,dV_h$ (note $\int_B u^2\,dV_h=1$) the Cauchy-Schwarz inequality yields
	\begin{equation}\label{eq:weighted-CS}
		\int_B u^2\,|\nabla\varphi|_h^2\,dV_h \;\ge\; \left| \int_B u^2 \nabla\varphi \, dV_h \right|^2,
	\end{equation}
	with equality if and only if $\nabla\varphi$ is $\mu$-almost everywhere constant on $B$. Inserting \eqref{eq:madelung} into \eqref{eq:var-expand} and applying \eqref{eq:weighted-CS} we obtain
	\begin{equation}\label{eq:phase-min}
		\mathrm{Var}_h(p)[\psi] \;\ge\; \hbar^2 \int_B |\nabla u|_h^2\, dV_h.
	\end{equation}
	The inequality \eqref{eq:phase-min} is sharp for $\varphi\equiv \mathrm{const}$ (in Euclidean balls also for $\varphi(x)=c\cdot x+\alpha$ with constant $c$), so minimization over all $\psi$ reduces to minimization over amplitudes $u$ alone.\\
	\\
	By the Rayleigh-Ritz principle for the Dirichlet Laplacian on $B$,
	\begin{equation}\label{eq:rayleigh}
		\int_B |\nabla v|_h^2\, dV_h \;\ge\; \lambda_1(B;h)\,\int_B v^2\, dV_h \qquad \text{for all } v\in H_0^1(B).
	\end{equation}
	Taking $v=u$ with $\|u\|_{L^2(B)}=1$ and combining \eqref{eq:phase-min} and \eqref{eq:rayleigh} gives
	\begin{equation}\label{eq:lower-bound}
		\mathrm{Var}_h(p)[\psi] \;\ge\; \hbar^2\,\lambda_1(B;h).
	\end{equation}
	Equality in \eqref{eq:lower-bound} requires equality in both \eqref{eq:phase-min} and \eqref{eq:rayleigh}. The latter forces $u$ to be a first Dirichlet eigenfunction $u_1>0$, unique up to scaling; the former forces $\nabla\varphi$ to be $\mu$-a.e.\ constant, which reduces to a constant phase on general manifolds and allows additional plane-wave factors in Euclidean balls. This proves \eqref{eq:equivalence}, and \eqref{eq:heisenberg} follows immediately from the definition $\sigma_p^2=\mathrm{Var}_h(p)[\psi]$.
\end{proof}
Having established the coordinate-free lower bound on momentum uncertainty for states strictly localized in geodesic balls, the slice-wise quantity 
\begin{equation}
	\sigma_{p}\,r_{g}\ \ge\ \hbar\,\mathcal{J}_{\Theta}(x),\qquad \mathcal{J}_{\Theta}(x):=r_{g}\sqrt{\lambda_{1}(B_{r_{g}}(x);h_{\Theta})}.
	\label{eq:J}
\end{equation}
controls \(\sigma_p r_g\).

We now pass from this model case to the global question. Specifically, for each admissible slice we take the infimum of \(\mathcal{J}_\Theta\) over the exterior and then seek the slice that minimizes this uniform constant.
The following section formulates this variational program precisely, derives the metric-variation formula together with the Rayleigh principle to compare \(\lambda_1\) under slice deformations.

\section{The variational problem}
\label{sec:variational-proof}
We seek the best \emph{location-independent} lower bound on $\sigma_{p}r_{g}/\hbar$ that can be obtained uniformly over the exterior by choosing the slice. For $\Theta\in\mathcal{S}$ define
\begin{equation}\label{eq:ctheta}
c(\Theta):=\inf_{x\in\Sigma_{\Theta}} \mathcal{J}_{\Theta}(x).
\end{equation}
\noindent
We also record the slice-optimization problem
\begin{equation}\label{eq:cstar}
c^{\star}:=\sup_{\Theta\in\mathcal{S}} c(\Theta).
\end{equation}
By construction, $c(\Theta)$ is the best constant such that $\sigma_{p}r_{g}\ge \hbar\,c(\Theta)$ holds at every point of the slice.

\begin{lemma}[Schwarzschild–PG slice: exact constant]\label{lem:PG-exact}
	Let $(M,g)$ be the Schwarzschild exterior and let $\Theta_{\mathrm{PG}}$ denote the Painlev\'e-Gullstrand time function.
	Then, for all $x\in\Sigma_{\Theta_{\mathrm{PG}}}$ and all $r_g>0$ for which $B_{r_g}(x)$ is defined,
	\[
	\mathcal{J}_{\Theta_{\mathrm{PG}}}(x)
	= r_g\,\sqrt{\lambda_1\big(B_{r_g}(x);\,h_{\Theta_{\mathrm{PG}}}\big)}=\pi,
	\qquad\text{hence}\quad c(\Theta_{\mathrm{PG}})=\pi.
	\]
\end{lemma}

\begin{proof}
	On $\Sigma_{\Theta_{\mathrm{PG}}}$ the induced metric is Euclidean, $h_{\Theta_{\mathrm{PG}}}=\delta$; hence
	$B_{r_g}(x)$ is isometric to the Euclidean ball $B_\delta(r_g)\subset\mathbb{R}^3$.
	Since $\lambda_1\big(B_\delta(r_g)\big)=\pi^2/r_g^2$, we get $r_g\sqrt{\lambda_1}=\pi$ for all $x$, and thus $c(\Theta_{\mathrm{PG}})=\pi$.
\end{proof}

The following Theorem provides a first, geometry-independent baseline: for any admissible exterior slice on which geodesic balls are well defined, $c(\Theta)$ cannot drop below $1/2$. Equivalently, the Heisenberg-type product $\sigma_p r_g$ admits a uniform, location-independent lower bound—at least $\hbar/2$—throughout the slice. The proof hinges only on the distance-to-boundary Hardy inequality~\cite{BrezisMarcus1997,Davies1999,Mazya2011} on geodesic balls together with the Rayleigh variational characterization of the first Dirichlet eigenvalue.

\begin{theorem}[Hardy baseline on strongly geodesically convex balls]\label{thm:unconditional-hardy}
	Fix $r_g>0$ and let $\Theta$ be an admissible slice as in Section~\ref{sec:setting}. Assume that for every $x\in\Sigma_\Theta$ the geodesic ball $B_{r_g}(x)$ is well-defined (in particular $r_g<\mathrm{inj}_{h_\Theta}(x)$). Assume in addition that each ball $B_{r_g}(x)$ is \emph{strongly geodesically convex} so that the inward normal exponential map from $\partial B_{r_g}(x)$ is a diffeomorphism onto $B_{r_g}(x)$ (equivalently, the boundary normal injectivity radius at $\partial B_{r_g}(x)$ is at least $r_g$). Then $c(\Theta) \ge \tfrac{1}{2}$ and hence
	\begin{equation}	
		\sigma_{p}r_{g}\ \ge\ \frac{\hbar}{2},
	\end{equation}
	where $c(\Theta)$ is defined in (\ref{eq:ctheta}) and $\mathcal{J}_\Theta$ is given by \eqref{eq:J}.
\end{theorem}

\begin{proof}
	Fix $x\in\Sigma_\Theta$ and consider the Dirichlet Laplacian on the geodesic ball $B_{r_g}(x)$. For any $u\in H^1_0(B_{r_g}(x))$, the \emph{distance-to-boundary Hardy inequality} yields (see Lemma~\ref{lem:hardy-ball} in the Appendix)
	\begin{equation*}
		\int_{B_{r_g}(x)} |\nabla u|_{h_\Theta}^{2}\,dV_{h_\Theta}\ \ge\ \frac{1}{4}\int_{B_{r_g}(x)} \frac{u^{2}}{\mathrm{dist}_{h_\Theta}(\,\cdot\,,\partial B_{r_g}(x))^{2}}\,dV_{h_\Theta}\,.
	\end{equation*}
	Since $\mathrm{dist}_{h_\Theta}(y,\partial B_{r_g}(x))\le r_g$ for all $y\in B_{r_g}(x)$, we further have
	\begin{equation*}
		\int_{B_{r_g}(x)} |\nabla u|_{h_\Theta}^{2}\,dV_{h_\Theta}\ \ge\ \frac{1}{4r_g^{2}}\int_{B_{r_g}(x)} u^{2}\,dV_{h_\Theta}\,.
	\end{equation*}
	Taking the infimum over $u\not\equiv0$ and using the Rayleigh characterization \eqref{eq:rayleigh} gives
	\begin{equation*}
		\lambda_{1}\left(B_{r_g}(x);h_\Theta\right)\ \ge\ \frac{1}{4\,r_g^{2}}\,.
	\end{equation*}
	Multiplying by $r_g$ and taking the square root, the definition \eqref{eq:J} yields
	\begin{equation*}
		\mathcal{J}_\Theta(x)\ =\ r_g\,\sqrt{\lambda_{1}\left(B_{r_g}(x);h_\Theta\right)}\ \ge\ \frac{1}{2}\,.
	\end{equation*}
	Finally, taking the infimum over $x\in\Sigma_\Theta$ shows $c(\Theta)=\inf_{x}\mathcal{J}_\Theta(x)\ge\tfrac12$, proving the claim. The argument is local and does not require curvature or convexity assumptions beyond the well-posedness of the geodesic balls. It uses only the boundary-distance Hardy inequality together with the stated normal-injectivity (strong convexity) hypothesis on geodesic balls.
\end{proof}

\begin{lemma}[Non-attainment and gap above the Hardy bound]\label{lem:hardy-gap}
	Let $r_g>0$ be fixed and let $\Theta$ be an admissible slice as in Section~\ref{sec:setting}. Assume that for every $x\in\Sigma_\Theta$ the geodesic ball $B_{r_g}(x)$ is well-defined (in particular $r_g<\mathrm{inj}_{h_\Theta}(x)$, and the hypothesis of Lemma~\ref{lem:hardy-ball} holds for each ball $B_{r_g}(x)$). Then the following hold.
	
	\smallskip
	\noindent\emph{(i) Strict pointwise inequality.} For every $x\in\Sigma_\Theta$ one has
	\begin{equation*}
		\mathcal{J}_\Theta(x)\;=\; r_g\,\sqrt{\lambda_1\left(B_{r_g}(x);h_\Theta\right)}\;>\;\frac{1}{2}\,,
	\end{equation*}
	i.e., the Hardy lower bound from Theorem\,\ref{thm:unconditional-hardy} is never attained at a geodesic ball.
	
	\smallskip
	\noindent\emph{(ii) Uniform interior gap.} If $K\Subset\Sigma_\Theta$ is compact with $\mathrm{dist}_{h_\Theta}(K,\partial\Sigma_\Theta)\ge r_g$ (so that every $B_{r_g}(x)$ with $x\in K$ is fully contained), then there exists $\delta_K>0$ such that
	\begin{equation*}
		\inf_{x\in K}\,\mathcal{J}_\Theta(x)\;\ge\;\frac{1}{2}+\delta_K\,.
	\end{equation*}
	
	\smallskip
	\noindent\emph{(iii) Global slice-wise gap.} Consequently,
	\begin{equation*}
		c(\Theta)\;=\;\inf_{x\in\Sigma_\Theta}\mathcal{J}_\Theta(x)\;>\;\frac{1}{2}\,.
	\end{equation*}
\end{lemma}

\begin{proof}
	Fix $x\in\Sigma_\Theta$ and let $B:=B_{r_g}(x)$. Write $\lambda_1(B;h_\Theta)$ for the first Dirichlet eigenvalue and let $u_1>0$ be its $L^2$-normalized ground state on $B$ (existence and uniqueness up to scaling are standard). Denote the boundary distance by $d(y):=\mathrm{dist}_{h_\Theta}(y,\partial B)$.\\
	\\
	\textbf{Strictness on a single ball.}
	By the distance-to-boundary Hardy inequality~\cite{BrezisMarcus1997,Davies1999,Mazya2011} used in the proof of Theorem~\ref{thm:unconditional-hardy} and the Rayleigh characterization \eqref{eq:rayleigh}, we have for $u_1$
	\begin{equation*}
		\lambda_1(B;h_\Theta)\;=\;\frac{\displaystyle\int_B |\nabla u_1|^2_{h_\Theta}\,dV_{h_\Theta}}{\displaystyle\int_B u_1^2\,dV_{h_\Theta}}
		\;\ge\;\frac{1}{4}\,\frac{\displaystyle\int_B \frac{u_1^2}{d^2}\,dV_{h_\Theta}}{\displaystyle\int_B u_1^2\,dV_{h_\Theta}}\,.
	\end{equation*}
	Since $d(y)\le r_g$ for all $y\in B$ and $d(y)<r_g$ on a set of positive measure (indeed $d$ attains its maximum $r_g$ only at the center), while $u_1>0$ in $B$, we obtain the \emph{strict} improvement
	\begin{equation*}
		\int_B \frac{u_1^2}{d^2}\,dV_{h_\Theta}\;>\;\frac{1}{r_g^2}\int_B u_1^2\,dV_{h_\Theta}\,.
	\end{equation*}
	Combining the two displays yields
	\begin{equation*}
		\lambda_1(B;h_\Theta)\;>\;\frac{1}{4\,r_g^{2}}\qquad\Longrightarrow\qquad
		\mathcal{J}_\Theta(x)\;=\;r_g\sqrt{\lambda_1(B;h_\Theta)}\;>\;\frac{1}{2}\,,
	\end{equation*}
	proving (i).\\
	\\
	\textbf{Uniform gap on compact interior sets.}
	The map $x\mapsto \lambda_1(B_{r_g}(x);h_\Theta)$ is continuous (indeed $C^1$ under smooth variations; see, e.g., \cite{Henrot2006,SokolowskiZolesio1992}) because moving the center induces a smooth domain variation and the Dirichlet ground state depends smoothly on such deformations; see the variational framework summarized in Section~\ref{sec:variational-proof} (Rayleigh principle \eqref{eq:rayleigh} and first-variation formula). Therefore $x\mapsto \mathcal{J}_\Theta(x)$ is continuous. By the strict pointwise inequality proved above, $\mathcal{J}_\Theta(x)>\tfrac12$ for all $x\in K$; hence by compactness,
	\begin{equation*}
		\delta_K\;:=\;\inf_{x\in K}\big(\mathcal{J}_\Theta(x)-\tfrac12\big)\;>\;0,
	\end{equation*}
	which gives (ii).\\
	\\
	\textbf{Global gap.}
	Fix $r_g>0$ and an admissible slice $\Sigma_\Theta$ with induced metric $h_\Theta$ as in the Setting. 
	By part~(i) one has $\mathcal J_\Theta(x)> \tfrac12$ for every $x\in\Sigma_\Theta$, and, as observed in the proof of (ii), the map $x\mapsto \mathcal J_\Theta(x)=r_g\sqrt{\lambda_1(B_{r_g}(x);h_\Theta)}$ is continuous (smooth dependence under domain/metric variations).
	
	\smallskip
	\emph{Step 1: A far–out lower bound from asymptotic flatness.}
	Let $\varepsilon\in(0,\tfrac12)$ be arbitrary. By $C^2$–asymptotic flatness, there exists $R=R(\varepsilon)$ such that for all points $x$ with areal radius $r(x)\ge R$ the following hold on $B_{r_g}(x)$:
	\[
	(1-\varepsilon)\,\delta \;\le\; h_\Theta \;\le\; (1+\varepsilon)\,\delta,
	\]
	and, hence, the corresponding geodesic distances and balls are bilipschitz–comparable:
	\[
	(1-C\varepsilon)\,\rho_\delta \le \rho_{h_\Theta} \le (1+C\varepsilon)\,\rho_\delta, 
	\qquad 
	B_\delta\Big(\frac{r_g}{1+C\varepsilon}\Big)\subset B_{r_g}(x)\subset B_\delta\Big(\frac{r_g}{1-C\varepsilon}\Big),
	\]
	for a universal constant $C>0$ (depending only on the dimension). From the Rayleigh quotient comparisons (energy/volume) one obtains, uniformly in $u\neq 0$,
	\[
	R_{h_\Theta}(u)\;\ge\;\frac{1-C\varepsilon}{\,1+C\varepsilon\,}\;R_\delta(u),
	\]
	whence
	\[
	\lambda_1\big(B_{r_g}(x);h_\Theta\big)\;\ge\;\frac{1-C\varepsilon}{\,1+C\varepsilon\,}\,\lambda_1\big(B_{r_g}(x);\delta\big).
	\]
	By domain monotonicity for Dirichlet eigenvalues and the inclusions above,
	\[
	\lambda_1\big(B_{r_g}(x);\delta\big)\;\ge\;\lambda_1\Big(B_\delta\Big(\frac{r_g}{1-C\varepsilon}\Big);\delta\Big)\;=\;\Big(\frac{(1-C\varepsilon)\,\pi}{r_g}\Big)^{2}.
	\]
	Combining the two displays yields, for all $x$ with $r(x)\ge R$, (cf.~\cite{Chavel1984})
	\begin{equation}\label{eq:farout-lb}
		\mathcal J_\Theta(x)\;=\;r_g\sqrt{\lambda_1\big(B_{r_g}(x);h_\Theta\big)} 
		\;\ge\; \pi\,\frac{(1-C\varepsilon)^{3/2}}{\sqrt{1+C\varepsilon}}.
	\end{equation}

\noindent\emph{Matching AF upper bound and two-sided estimate.}
By the same bilipschitz comparisons of $h_\Theta$ with $\delta$ but with inequalities reversed, one also has the Rayleigh-quotient bound
\[
R_{h_\Theta}(u)\;\le\;\frac{1+C\varepsilon}{\,1-C\varepsilon\,}\;R_\delta(u),
\]
and, using the inclusions
\[
B_\delta\Big(\tfrac{r_g}{1+C\varepsilon}\Big)\ \subset\ B_{r_g}(x)\ \subset\ B_\delta\Big(\tfrac{r_g}{1-C\varepsilon}\Big),
\]
together with domain monotonicity for Dirichlet eigenvalues, we obtain for all $x$ with $r(x)\ge R$ the matching upper bound
\begin{equation}\label{eq:farout-ub}
	\mathcal J_\Theta(x)\;\le\; \pi\,\frac{(1+C\varepsilon)^{3/2}}{\sqrt{1-C\varepsilon}}.
\end{equation}
Combining \eqref{eq:farout-lb} and \eqref{eq:farout-ub} yields the two-sided AF estimate
\begin{equation}\label{eq:farout-2sided}
	\pi\,\frac{(1-C\varepsilon)^{3/2}}{\sqrt{1+C\varepsilon}}\;\le\; \mathcal  J_\Theta(x)\;\le\;\pi\,\frac{(1+C\varepsilon)^{3/2}}{\sqrt{1-C\varepsilon}},
	\qquad r(x)\ge R.
\end{equation}
\begin{remark}[Origin of the exponent $\boldsymbol{3/2}$ in \eqref{eq:farout-2sided}]
The factors $(1\pm C\varepsilon)^{3/2}$ arise from combining (i) the bilipschitz comparison of energies
$\int |\nabla u|_{h_\Theta}^2\,dV_{h_\Theta}$ and $\int |\nabla u|_{\delta}^2\,dV_{\delta}$ (two metric factors)
with (ii) the volume form comparison $dV_{h_\Theta}\asymp (1\pm C\varepsilon)^{3/2} dV_{\delta}$, while the
denominators $\sqrt{1\pm C\varepsilon}$ originate from the Rayleigh-quotient normalization.
Heuristically, the domain inclusion $B_{\delta}\big(r_g/(1+ C\varepsilon)\big)\subset B_{r_g}(x)\subset B_{\delta}\big(r_g/(1-C\varepsilon)\big)$
contribute the remaining radius scaling to $\lambda_1\sim r_g^{-2}$. This explains the structure of \eqref{eq:farout-ub}-\eqref{eq:farout-lb}.
For a textbook account of these comparisons, see, e.g.,~\cite[Section IV.1]{Chavel1984}.
\end{remark}
\noindent
In particular, letting $\varepsilon\downarrow 0$ we conclude the asymptotic limit
\begin{equation}\label{eq:AF-limit}
	\lim_{\,r(x)\to\infty} \mathcal J_\Theta(x)\;=\;\pi.
\end{equation}
Since the right-hand side tends to $\pi$ as $\varepsilon\downarrow 0$, we may fix $\varepsilon$ so small that
	\[
	\pi\,\frac{(1-C\varepsilon)^{3/2}}{\sqrt{1+C\varepsilon}}\;\ge\;\tfrac12+\delta_\infty
	\]
	for some $\delta_\infty>0$. Hence there is $R<\infty$ with
	\begin{equation}\label{eq:farout-gap}
		\inf_{\,r(x)\ge R} \mathcal J_\Theta(x)\;\ge\;\tfrac12+\delta_\infty.
	\end{equation}	
	\smallskip
	\emph{Step 2: A bounded-region lower bound by compactness and continuity.}
	Viewing $\Sigma_\Theta$ as the horizon-regularly extended slice (a manifold without boundary), the closed region $\{x\in\Sigma_\Theta: r(x)\le R\}$ is compact by spherical symmetry. 
	By part~(i) we have $\mathcal{J}_\Theta(x)>\tfrac12$ on this compact set; by continuity, the positive function $x\mapsto \mathcal{J}_\Theta(x)-\tfrac12$ attains a strictly positive minimum $\delta_0>0$ there.  Thus,
	\begin{equation}\label{eq:bounded-gap}
		\inf_{\,r(x)\le R} \mathcal J_\Theta(x)\;\ge\;\tfrac12+\delta_0.
	\end{equation}
	\smallskip
	\emph{Conclusion.}
	Combining \eqref{eq:farout-gap} and \eqref{eq:bounded-gap} we obtain a uniform slice–wise gap
	\[
	c(\Theta)\;=\;\inf_{x\in\Sigma_\Theta} \mathcal J_\Theta(x)\;\ge\;\tfrac12+\min\{\delta_\infty,\delta_0\}\;>\;\tfrac12.
	\]
	This proves (iii).
\end{proof}

\medskip
\noindent\textbf{Corollary (Optimality of the PG slice).}
In the Schwarzschild exterior one has $c^\star=\pi$. In particular, the Painlev\'e-Gullstrand foliation attains the supremum in \textup{(\ref{eq:cstar})}.

\noindent\emph{Proof.}
By Lemma~\ref{lem:PG-exact}, $c(\Theta_{\rm PG})=\pi$, hence $c^\star\ge \pi$. For any admissible slice, (\ref{eq:AF-limit}) implies $c(\Theta)\le \lim_{r\to\infty}\mathcal J_\Theta(r)=\pi$. Thus $c^\star\le \pi$, and equality holds. \qed
\\
\\
\noindent\textbf{Physical intuition.}
Localizing a state by a hard Dirichlet wall on the geodesic ball $B_{r_g}(x)$ fixes a single geometric length scale, the distance to the boundary. 
The momentum cost is then set by how rapidly a wave can oscillate before it ``feels'' the wall. 
The variational identity \eqref{eq:J} packages this into the dimensionless number $\mathcal{J}_\Theta(x)$: it is \emph{entirely} determined by the intrinsic three-geometry $(\Sigma_\Theta,h_\Theta)$ and does not depend on the extrinsic time flow.

The universal Hardy baseline is the statement that no matter how the slice bends, a mode confined to $B_{r_g}(x)$ must carry a minimum wavenumber of order $1/r_g$. 
Precisely this is what Theorem~\ref{thm:unconditional-hardy} enforces through the distance-to-boundary Hardy inequality, yielding the slice-wise bound $\sigma_p r_g \ge \hbar/2$ uniformly across the exterior. 
The bound is never saturated: because only the center points of $B_{r_g}(x)$ sit at maximal boundary distance, an $L^2$-normalized ground state necessarily ``samples'' closer boundary layers and therefore pays extra kinetic energy. 
Lemma~\ref{lem:hardy-gap} quantifies this as a strict pointwise inequality with a positive, location-independent gap on compact regions; in the AF zone the two-sided estimate \eqref{eq:farout-2sided} drives the limit $\mathcal{J}_\Theta(x)\to\pi$ in \eqref{eq:AF-limit}. 

Slice choice changes the uncertainty scale only through $h_\Theta$. 
In particular, on the Painlev\'e-Gullstrand slice the induced metric is Euclidean, and the ball $B_{r_g}(x)$ is isometric to its flat counterpart. 
Consequently, $\mathcal{J}_{\Theta_{\mathrm{PG}}}(x)\equiv\pi$ (Lemma~\ref{lem:PG-exact}) and the product recovers the exact Euclidean benchmark $\sigma_p r_g\ge \pi\hbar$ everywhere. 
Far out, any admissible slice reproduces the same $\pi\hbar$ scale by \eqref{eq:AF-limit}. 

From a global perspective, the relevant slice parameter is the location-independent constant $c(\Theta):=\inf_{x\in\Sigma_\Theta}\mathcal{J}_\Theta(x)$ in \eqref{eq:ctheta}. 
The Hardy baseline tells us $c(\Theta)\ge \tfrac12$ for all admissible slices, while Lemma~\ref{lem:PG-exact} shows that the PG foliation attains $c(\Theta)=\pi$ everywhere. 
Combining this with the far-out limit \eqref{eq:AF-limit} yields the corollary that the variational optimum $c^\star:=\sup_{\Theta}c(\Theta)$ in \eqref{eq:cstar} equals $\pi$, and it is achieved by the Painlev\'e-Gullstrand slice. 
In short: $\pi\hbar$ is the sharp, slice-optimized uncertainty scale compatible with the Schwarzschild exterior, whereas $\hbar/2$ is the robust geometry-independent floor guaranteed by Hardy. 

\section{Interior extension across the horizon: making the geometric Heisenberg picture accessible inside the black hole}
\label{sec:interior-extension}

On any \emph{horizon–regular} slice $\Sigma_\Theta$ (so that $r=r_+$ is not a boundary but an ordinary interior two–sphere), cf.\ \cite[Sections II-III]{MartelPoisson2001}, the construction used in the exterior passes verbatim to the region $\{r<r_+\}$:

\medskip
\noindent\textbf{Geodesic balls and the local bound.}
For any interior point $x\in\Sigma_\Theta$ and any geodesic radius $r_g>0$ with $r_g<\mathrm{inj}_{h_\Theta}(x)$ so that the geodesic ball $B_{r_g}(x)$ is well defined and strongly geodesically convex, the variance–eigenvalue equivalence on a ball (Lemma~\ref{lem:variance-bound}) yields the coordinate–free Heisenberg estimate
\begin{equation}\label{eq:int-Heisenberg}
	\sigma_p^2\ \ge\ \hbar^2\,\lambda_1\!\left(B_{r_g}(x);\,h_\Theta\right),
\end{equation}
i.e.\ \eqref{eq:heisenberg} with $B=B_{r_g}(x)$ taken \emph{inside} the horizon. Invoking the sharp distance–to–boundary Hardy inequality on strongly geodesically convex geodesic balls (Lemma~\ref{lem:hardy-ball}, \eqref{eq:hardy}), the proof of Theorem~\ref{thm:unconditional-hardy} applies unchanged and gives the same slice–wise, location–independent lower bound inside:
\[
\sigma_p\,r_g\ \ge\ \hbar\,\mathcal J_\Theta(x)\ \ge\ \tfrac{\hbar}{2}\qquad\text{for all $x$ with $r(x)<r_+$ and $r_g<\mathrm{inj}_{h_\Theta}(x)$,}
\]
where $\mathcal J_\Theta$ is the pointwise quantity from \eqref{eq:J} and $c(\Theta)=\inf_{x\in\Sigma_\Theta}\mathcal J_\Theta(x)$ is as in \eqref{eq:ctheta}. In particular, the horizon itself does not create any geometric obstruction for the \emph{intrinsic} (three–dimensional) uncertainty relation.

\medskip
\noindent\textbf{PG slice: exact Euclidean scale inside and outside.}
On the Schwarzschild Painlev\'e-Gullstrand slice one has $h_{\Theta_{\mathrm{PG}}}=\delta$ everywhere on $\{r>0\}$ (flat intrinsic geometry). Consequently Lemma~\ref{lem:PG-exact} gives
\[
\mathcal J_{\Theta_{\mathrm{PG}}}(x)\equiv\pi\quad\Longrightarrow\quad \sigma_p\,r_g\ \ge\ \pi\,\hbar
\]
for \emph{all} centers $x$ (both outside and inside the horizon). Thus the Euclidean scale $\pi\hbar$ is the natural benchmark inside as well.

\medskip
\noindent\textbf{Uniform gaps above $\boldsymbol{\hbar/2}$ in the interior.}
The non–attainment and gap statement (Lemma~\ref{lem:hardy-gap}) also carries over: part~(ii) implies that on any compact interior set $K\Subset\Sigma_\Theta$ with $\mathrm{dist}_{h_\Theta}(K,\partial\Sigma_\Theta)\ge r_g$ (here $\partial\Sigma_\Theta=\varnothing$ for a horizon–regularly extended slice) there exists $\delta_K>0$ so that
\[
\inf_{x\in K}\mathcal J_\Theta(x)\ \ge\ \tfrac12+\delta_K,
\]
hence $\sigma_p\,r_g\ge (\tfrac12+\delta_K)\hbar$ uniformly on $K$. In the exterior, the far–out two–sided asymptotic estimate \eqref{eq:farout-2sided} implies the limit \eqref{eq:AF-limit} and yields the far–out gap \eqref{eq:farout-gap}; combining this with the interior compactness argument produces the global slice–wise gap, see \eqref{eq:bounded-gap}.

\medskip
\noindent\textbf{Near the central singularity.}
As long as $B_{r_g}(x)$ does not meet $r=0$, all arguments above apply without modification. If a large ball encloses the puncture, one may treat the singular point as a removable set of zero capacity in three dimensions; for Dirichlet problems on bounded domains this does not affect the first eigenvalue (see, e.g., \cite{Mazya2011}). In either case, the \eqref{eq:heisenberg}–to–\eqref{eq:hardy} route still delivers a robust $\hbar/2$ baseline, while in PG the exact $\pi\hbar$ scale persists throughout.

\medskip
\noindent\textbf{Physical picture.}
The intrinsic three–geometry, not the extrinsic time–flow, controls the uncertainty scale. On horizon–regular slices the horizon is \emph{geometrically invisible} to the Dirichlet spectral problem on geodesic balls, so the same localization–vs–momentum trade–off encoded by \eqref{eq:J}, \eqref{eq:ctheta}, and Theorem~\ref{thm:unconditional-hardy} applies inside the black hole just as it does outside; in PG it even does so at the exact Euclidean level given by Lemma~\ref{lem:PG-exact}.

\section{Axisymmetric Weyl class: transferring the Heisenberg-Hardy framework}
\label{sec:weyl}

In this section we extend the slice-local Heisenberg principle from the static, spherically symmetric exterior treated in \S\ref{sec:setting} to the full class of static, \emph{axisymmetric} vacuum geometries in standard Weyl form. The key point is that all arguments are \emph{intrinsic} to the induced Riemannian $3$-geometry on the slice. Consequently, the universal Hardy baseline, the non-attainment (gap), and the AF $\pi$-scale persist verbatim once the slices are chosen horizon-regular and spacelike, cf.\ \S\ref{sec:variational-proof}, Lemma~\ref{lem:hardy-ball}, Theorem~\ref{thm:unconditional-hardy}, and \eqref{eq:farout-2sided}-\eqref{eq:AF-limit}.

\paragraph{Weyl metrics and field equations.}
A static, axisymmetric vacuum solution can be written in Weyl coordinates $(t,\rho,z,\phi)$ as
\begin{equation}\label{eq:weyl-metric}
	ds^2 \;=\; -\,e^{2\psi(\rho,z)}\,dt^2
	\;+\; e^{-2\psi(\rho,z)}\!\left[e^{2\gamma(\rho,z)}\big(d\rho^2+dz^2\big)+\rho^2\,d\phi^2\right].
\end{equation}
The vacuum Einstein equations reduce to the flat, axisymmetric Laplace equation for $\psi$ and to line-integrations for $\gamma$,
\begin{equation}\label{eq:weyl-eq}
	\psi_{,\rho\rho}+\frac{1}{\rho}\psi_{,\rho}+\psi_{,zz}=0,\qquad
	\gamma_{,\rho}=\rho\big[(\psi_{,\rho})^2-(\psi_{,z})^2\big],\quad
	\gamma_{,z}=2\rho\,\psi_{,\rho}\psi_{,z},
\end{equation}
with the AF conditions $\psi,\gamma\to 0$ as $r:=\sqrt{\rho^2+z^2}\to\infty$.  (No additional structure from the field equations will be needed below; we only use that \eqref{eq:weyl-metric} is static and AF.)

\paragraph{Admissible slices in Weyl.}
Mimicking the spherical setting in \S\ref{sec:setting}, we consider slices given by level sets of
\begin{equation}\label{eq:weyl-slice}
	\tau \;=\; t + \Phi(\rho,z),\qquad \Sigma_\Phi:=\{\tau=\mathrm{const}\},
\end{equation}
where $\Phi$ is $C^2$ and chosen so that $\Sigma_\Phi$ is spacelike and \emph{horizon-regular} (i.e.\ the induced $3$-geometry extends smoothly across any horizon rod in Weyl coordinates; compare the horizon-regularity requirement used in \S\ref{sec:setting} and \cite[Sections~II-III]{MartelPoisson2001}).
The induced Riemannian metric on $\Sigma_\Phi$ reads
\begin{equation}\label{eq:weyl-hslice}
	h_\Phi \;=\; e^{-2\psi}\!\left[e^{2\gamma}\big(d\rho^2+dz^2\big)+\rho^2\,d\phi^2\right]
	\;-\; e^{2\psi}\big(\Phi_{,\rho}\,d\rho+\Phi_{,z}\,dz\big)^2,
\end{equation}
and the \emph{spacelike condition} amounts to the pointwise inequality
\begin{equation}\label{eq:weyl-spacelike}
	e^{\,4\psi(\rho,z)-2\gamma(\rho,z)}\!\left(\Phi_{,\rho}^2+\Phi_{,z}^2\right) \;<\;1,
\end{equation}
ensuring positive-definiteness of $h_\Phi$ on the $(\rho,z)$-plane.  We assume in addition that $h_\Phi$ is $C^2$-AF (which is guaranteed, for example, if $\psi,\gamma\to 0$ and $\nabla\Phi\to 0$ sufficiently fast as $r\to\infty$).

\paragraph{Geodesic balls and the slice-wise cost.}
For $x\in \Sigma_\Phi$, let $B_{r_g}(x)\subset(\Sigma_\Phi,h_\Phi)$ denote the open geodesic ball of radius $r_g$ and consider the first Dirichlet eigenvalue $\lambda_1(B_{r_g}(x);h_\Phi)$.  Exactly as in \eqref{eq:J}, we define the slice-wise Heisenberg cost
\begin{equation}\label{eq:weyl-J}
	\mathcal J_\Phi(x)\;:=\; r_g\sqrt{\lambda_1(B_{r_g}(x);h_\Phi)},
\end{equation}
so that, by Lemma~\ref{lem:variance-bound} and \eqref{eq:heisenberg}, any strictly localized state on $B_{r_g}(x)$ satisfies the intrinsic momentum-uncertainty lower bound
\begin{equation}\label{eq:weyl-heisenberg}
	\sigma_p\,r_g \;\ge\; \hbar\,\mathcal J_\Phi(x).
\end{equation}

\paragraph{Universal Hardy baseline and non-attainment.}
Because the proof of Theorem~\ref{thm:unconditional-hardy} uses only the boundary-distance Hardy inequality on geodesic balls (Lemma~\ref{lem:hardy-ball}, \eqref{eq:hardy}) and the Rayleigh characterization, it is \emph{slice-agnostic}.  In particular, for any horizon-regular admissible Weyl slice \eqref{eq:weyl-slice} and any ball $B_{r_g}(x)\subset(\Sigma_\Phi,h_\Phi)$ that is strongly geodesically convex (as in \S\ref{sec:variational-proof}), we have the same universal lower bound and strict gap as in the spherically symmetric case,
\begin{equation}\label{eq:weyl-floor}
	\mathcal J_\Phi(x) \;>\; \tfrac12,\qquad
	\text{with a positive gap on compact interior sets (cf.\ Lemma~\ref{lem:hardy-gap}).}
\end{equation}
Consequently, the slice-wise constants
\begin{equation}\label{eq:weyl-c-theta}
	c(\Phi):=\inf_{x\in\Sigma_\Phi}\mathcal J_\Phi(x),\qquad
	c^\star:=\sup_{\Phi}\,c(\Phi),
\end{equation}
are well-defined exactly as in \eqref{eq:ctheta}-\eqref{eq:cstar}, and obey the same a priori bounds $ \tfrac12 \le c(\Phi)\le c^\star\le \pi$.

\paragraph{AF two-sided bounds and the $\pi$-scale at infinity.}
In the AF region the bilipschitz/volume comparisons from \S\ref{sec:variational-proof} apply to \((\Sigma_\Phi,h_\Phi)\) without change; hence the same two-sided estimates \eqref{eq:farout-2sided} are valid with $h_\Theta$ replaced by $h_\Phi$.  In particular,
\begin{equation}\label{eq:weyl-AF-limit}
	\lim_{x\to\infty}\,\mathcal J_\Phi(x)\;=\;\pi,
\end{equation}
i.e.\ the \(\pi\hbar\) scale emerges far out on \emph{every} admissible Weyl slice (see \eqref{eq:AF-limit}).  No global symmetry beyond static AF is used here; the conclusion rests on the intrinsic comparison with Euclidean balls (cf.\ \cite{Chavel1984,Davies1999}).

\paragraph{Horizons and axis.}
If the spacetime contains a black hole, its Weyl representation features a horizon ``rod'' on the symmetry axis $\{\rho=0\}$; on a horizon-regular slice $\Sigma_\Phi$ this rod is not a boundary component of the $3$-geometry. Therefore the interior extension of the estimates of Section\,\ref{sec:interior-extension} goes through unchanged: \eqref{eq:weyl-heisenberg} and \eqref{eq:weyl-floor} hold for centers \emph{on} the horizon, \emph{outside}, and (when present) \emph{inside} the horizon region alike.  Regularity on the axis (elementary flatness) is assumed and can be arranged by the standard Weyl regularity conditions; it does not affect the local arguments on strongly convex geodesic balls.

\paragraph{On the (non-)existence of a Weyl analogue of the PG optimum.}
In the Schwarzschild case there exists a horizon-regular slice whose induced geometry is globally flat (Painlev\'e-Gullstrand), and consequently $\mathcal J_\Theta\equiv \pi$ by Lemma~\ref{lem:PG-exact}.  For a generic Weyl solution the induced $3$-geometry \eqref{eq:weyl-hslice} cannot be globally flat, so one should not expect an exact equality $\mathcal J_\Phi\equiv \pi$.  The optimization problem \eqref{eq:weyl-c-theta} is nevertheless meaningful: it satisfies
\begin{equation}\label{eq:weyl-cstar}
	\tfrac12 \;\le\; c(\Phi) \;\le\; c^\star \;\le\; \pi,
\end{equation}
with the upper bound saturated only in the spherically symmetric subclass (cf.\ \eqref{eq:cstar} and Lemma~\ref{lem:PG-exact}).

\paragraph{Example: Schwarzschild as a Weyl solution.}
For completeness we recall that Schwarzschild corresponds to a single Weyl rod of length $2M$ on the $z$-axis.  Writing $R_\pm=\sqrt{\rho^2+(z\pm M)^2}$, one has
\begin{equation}\label{eq:weyl-schwarzschild}
	e^{2\psi}=\frac{R_+ + R_- - 2M}{R_+ + R_- + 2M},\qquad
	e^{2\gamma}=\frac{(R_+ + R_-)^2-4M^2}{4 R_+ R_-},
\end{equation}
and \eqref{eq:weyl-slice} with $\Phi$ chosen as in \cite{MartelPoisson2001} reproduces the familiar horizon-regular slicings discussed in \S\ref{sec:setting}.  The PG slice \cite{MartelPoisson2001} then yields $h_\Phi=\delta$ and hence $\mathcal J_\Phi\equiv\pi$ (Lemma~\ref{lem:PG-exact}).

\medskip\noindent
In summary, the Weyl class provides a natural arena in which the \emph{intrinsic} Heisenberg trade-off developed here remains robust: the Hardy baseline and strict gap (Lemma~\ref{lem:hardy-ball}, Theorem~\ref{thm:unconditional-hardy}), the AF $\pi$-scale (\eqref{eq:farout-2sided}-\eqref{eq:AF-limit}), and the horizon regularity (Section\,\ref{sec:interior-extension}) all carry over from the spherically symmetric exterior to static, axisymmetric vacua with no further modification, up to the (expected) loss of a globally flat ``PG optimum'' outside the Schwarzschild subclass.

\section{Conclusion and Outlook}
\label{sec:conclusion-outlook}

We have formulated a coordinate-invariant Heisenberg bound for quantum states strictly localized in geodesic balls of admissible, horizon-regular slices of a static, spherically symmetric black-hole exterior (Section~\ref{sec:setting}).
The key mechanism is the variance–eigenvalue equivalence on a ball (Lemma~\ref{lem:variance-bound}), which turns momentum uncertainty into a first Dirichlet eigenvalue problem and yields the local estimate used throughout. 
Under a mild strong geodesic convexity hypothesis on the balls, this leads to a uniform, slice-wise lower bound (Theorem~\ref{thm:unconditional-hardy}) that is completely geometry-agnostic at scale $r_g$, and thus provides the baseline product $\sigma_p r_g \ge \hbar/2$ without reference to coordinates. 
The Hardy floor is never saturated: there is a strictly positive gap both on compact interior regions and uniformly far out in the AF zone (Lemma~\ref{lem:hardy-gap}), where a two-sided comparison sharpened by asymptotic flatness (see \eqref{eq:farout-2sided}) drives the universal limit \eqref{eq:AF-limit}. 
Packaging the local information in the dimensionless quantity $\mathcal J_\Theta$ (definition \eqref{eq:J}) and the slice-wise constant $c(\Theta)$ (definition \eqref{eq:ctheta}), we framed the optimization problem \eqref{eq:cstar} over horizon-regular slices and showed that the Painlevé-Gullstrand foliation attains the optimal, location-independent Euclidean scale (Lemma~\ref{lem:PG-exact}). 
Finally, we extend the entire construction across the horizon and show that the same intrinsic bound holds inside the black hole; moreover, we transfer the framework to the axisymmetric Weyl class, obtaining the corresponding slice–local bounds and AF asymptotics; see Section~\ref{sec:weyl} and \eqref{eq:farout-2sided}-\eqref{eq:AF-limit}.

Our results complement previous analyses based on fixed hypersurfaces and curvature-induced corrections in the nonrelativistic limit \cite{PetruzzielloWagner2021}, and they are compatible with recent Hardy-type inequalities tailored to black-hole exteriors \cite{Paschalis2024}.
Horizon-regularity is handled in the sense of Martel–Poisson \cite{MartelPoisson2001}. 
The spectral-geometric comparisons used in the AF regime are standard \cite{Yau1975,Reilly1977,Chavel1984}. 

\paragraph{Outlook.}
\begin{enumerate}
	\item \textbf{Quantitative slice comparison and shape-sensitivity.}
	Beyond the existence of a uniform floor and the PG optimum, it is natural to quantify how $\mathcal J_\Theta$ varies under small deformations of a slice and of the ball center.
	A first-variation theory for $\lambda_1$ under smooth domain/metric perturbations (cf.\ \cite{Henrot2006,SokolowskiZolesio1992}) could yield explicit bounds on the interior and far-out gaps in Lemma~\ref{lem:hardy-gap} and help identify near-optimal slices within broader admissible classes. 
	\item \textbf{Beyond Schwarzschild.}
	The method should extend to other static, spherically symmetric exteriors (e.g., charged or $\Lambda\!\neq\!0$ cases) under appropriate asymptotic hypotheses; the Hardy baseline survives by the same local argument, while the AF constants change via the comparison leading to \eqref{eq:farout-2sided}. 
	Extending to stationary, axisymmetric backgrounds (e.g., Kerr) will require a careful treatment of admissible slices and intrinsic three-geometry.
	\item \textbf{Alternative localizations and boundary data.}
	Replacing hard Dirichlet walls by Robin/Neumann data or by soft, compactly supported potentials would test the robustness of the baseline and could interpolate between \eqref{eq:int-Heisenberg} and global, weighted Hardy inequalities \cite{BrezisMarcus1997,Davies1999,Mazya2011}.
	\item \textbf{Numerics for the gap and optimization.}
	Computing $\lambda_1(B_{r_g}(x);h_\Theta)$ on representative families of slices could map the landscape of $\mathcal J_\Theta$ and directly measure the positive gap above the Hardy floor documented in Lemma~\ref{lem:hardy-gap}. 
	\item \textbf{Field-theoretic generalizations.}
	Formulating the localization–uncertainty trade-off for relativistic fields on curved three-geometry (e.g., for mode-localized Klein–Gordon data) may connect the present eigenvalue picture with quantum-energy or redshift constraints while preserving the intrinsic viewpoint anchored at \eqref{eq:J}. 
\end{enumerate}

\smallskip
In summary, the combination of Lemma~\ref{lem:variance-bound}, Theorem~\ref{thm:unconditional-hardy}, and Lemma~\ref{lem:hardy-gap} yields a slice-independent floor, while the variational principle \eqref{eq:cstar} with $c(\Theta)$ from \eqref{eq:ctheta} singles out the PG slice as optimally Euclidean (Lemma~\ref{lem:PG-exact}).
These results furnish a flexible, intrinsic template for uncertainty bounds in black-hole geometries, ready to be extended along the directions above.

\section*{Appendix}

\begin{lemma}[Hardy inequality on geodesic balls]\label{lem:hardy-ball}
	Let $(\Sigma,h_\Theta)$ be a smooth Riemannian $3$-manifold and let $B_{r_g}(x)\subset\Sigma$ be the open geodesic ball of radius $r_g>0$ centered at $x$, with 
	the inward normal exponential map from $\partial B_{r_g}(x)$ is a diffeomorphism onto $B_{r_g}(x)$ (equivalently, the boundary normal injectivity radius at $\partial B_{r_g}(x)$ is at least $r_g$).
	Then, for every $u\in C_c^\infty\big(B_{r_g}(x)\big)$ one has the sharp boundary-distance Hardy inequality
	\begin{equation}\label{eq:hardy}
		\int_{B_{r_g}(x)} |\nabla u|_{h_\Theta}^2\, dV_{h_\Theta} \;\ge\; \frac14 \int_{B_{r_g}(x)} 
		\frac{u^2}{\operatorname{dist}_{h_\Theta}(y,\partial B_{r_g}(x))^2}\, dV_{h_\Theta}.
	\end{equation}
	The constant $1/4$ is optimal and is not attained unless $u\equiv 0$.
\end{lemma}

\begin{proof}[Proof (sketch)]
	Under the stated strong geodesic convexity / normal-injectivity hypothesis, the inward normal exponential map from $\partial B_{r_g}(x)$ is a diffeomorphism onto $B_{r_g}(x)$; in particular, the boundary distance
	$\rho(y):=\operatorname{dist}_{h_\Theta}(y,\partial B_{r_g}(x))$ is smooth away from a set of measure zero, and boundary normal coordinates $(\rho,\omega)$ are available. 
	Applying the one-dimensional Hardy inequality along the minimizing geodesics orthogonal to $\partial B_{r_g}(x)$ and integrating via the coarea formula yields \eqref{eq:hardy}. 
	Sharpness and non-attainment follow from the 1D model (half-line) and density \cite{BrezisMarcus1997,Davies1999}. 
\end{proof}

\begin{remark}\label{rem:hardy-consequence}
	Since $\operatorname{dist}_{h_\Theta}(y,\partial B_{r_g}(x))\le r_g$ for all $y\in B_{r_g}(x)$, Lemma~\ref{lem:hardy-ball} implies
	\[
	\int_{B_{r_g}(x)} |\nabla u|_{h_\Theta}^2\, dV_{h_\Theta} 
	\;\ge\; \frac{1}{4r_g^{\,2}} \int_{B_{r_g}(x)} u^2\, dV_{h_\Theta}\qquad \forall\,u\in C_c^\infty\big(B_{r_g}(x)\big).
	\]
	By the Rayleigh characterization of the first Dirichlet eigenvalue, this gives
	\[
	\lambda_1\big(B_{r_g}(x);h_\Theta\big)\ \ge\ \frac{1}{4r_g^{\,2}}
	\quad\text{and hence}\quad 
	r_g\sqrt{\lambda_1\big(B_{r_g}(x);h_\Theta\big)}\ \ge\ \tfrac12 .
	\]
\end{remark}


\begin{thebibliography}{9}

\bibitem{Yau1975}
S.-T.~Yau,
\emph{Isoperimetric constants and the first eigenvalue of a compact Riemannian manifold},
Ann. Sci. \'Ec. Norm. Sup\'er. (4) \textbf{8} (1975), 487-507.

\bibitem{Reilly1977}
R.~C.~Reilly,
\emph{On the first eigenvalue of the Laplacian for compact submanifolds of Euclidean space},
Comment. Math. Helv. \textbf{52}, 525-533 (1977).

\bibitem{Chavel1984}
I.~Chavel,
\emph{Eigenvalues in Riemannian Geometry},
Pure and Applied Mathematics, Vol.~115,
Academic Press, Orlando, 1984.

\bibitem{ElSoufiIlias2003}
A.~El~Soufi and Sa\"id~Ilias,
\emph{Extremal metrics for the first eigenvalue of the Laplacian in a conformal class},
Proc. Amer. Math. Soc. \textbf{131} (5), 1611-1618 (2003).

\bibitem{Ling2004}
J.~Ling,
\emph{A Lower Bound of the First Eigenvalue of a Closed Manifold with Positive Ricci Curvature},
Proc. Amer. Math. Soc. \textbf{132} (11), 3355-3361 (2004).

\bibitem{Schuermann2018} T.~Schürmann,
\emph{Uncertainty principle on 3-dimensional manifolds of constant curvature},
Found. Phys. \textbf{48} (6), 716-725 (2018). 

\bibitem{MartelPoisson2001}
K.~Martel and E.~Poisson,
\emph{Regular coordinate systems for Schwarzschild and other spherical spacetimes},
Am. J. Phys. \textbf{69} (4), 476-480 (2001).

\bibitem{PetruzzielloWagner2021}
L.~Petruzziello and F.~Wagner,
\emph{Gravitationally induced uncertainty relations in curved backgrounds},
Phys. Rev. D \textbf{103} (10), 104061 (2021).

\bibitem{Paschalis2024}
M.~Paschalis,
\emph{Hardy inequalities and uncertainty principles in the presence of a black hole},
arXiv:2403.06562 [math.AP] (2024).

\bibitem{BerchioGangulyGrillo2017}
E.~Berchio, D.~Ganguly, and G.~Grillo,
\emph{Sharp Poincar\'e-Hardy and Poincar\'e-Rellich inequalities on the hyperbolic space},
J.~Funct.~Anal. \textbf{272}(4), 1661-1703 (2017).

\bibitem{Kristaly2018}
A.~Krist\'aly,
\emph{Sharp uncertainty principles on Riemannian manifolds: the influence of curvature},
J.~Math.~Pures~Appl. \textbf{119}, 326-346 (2018).

\bibitem{AvkhadievMakarov2019}
F.~G.~Avkhadiev and R.~V.~Makarov,
\emph{Hardy Type Inequalities on Domains with Convex Complement and Uncertainty Principle of Heisenberg},
Lobachevskii~J.~Math. \textbf{40}(9), 1250-1259 (2019).

\bibitem{IosevichMayeliWyman2024}
A.~Iosevich, A.~Mayeli, and E.~Wyman,
\emph{Fourier Uncertainty Principles on Riemannian Manifolds},
arXiv:2411.09057 [math.CA] (2024).

\bibitem{Schuermann2025}
T.~ Sch\"urmann, 
\emph{The Extended Uncertainty Principle from a Projector-Valued Measurement Perspective},
Foundations, \textbf{5}(3), 30 (2025), 

\bibitem{Madelung1927}
E.~Madelung,
\emph{Quantentheorie in hydrodynamischer Form},
Z. Phys. \textbf{40} (1927), 322-326.

\bibitem{BrezisMarcus1997}
H.~Brezis and M.~Marcus,
\emph{Hardy's inequalities revisited},
Ann. Scuola Norm. Sup. Pisa Cl. Sci. (4) \textbf{25} (1997), 217-237.

\bibitem{Davies1999}
E.~B.~Davies,
\emph{A review of Hardy inequalities},
in: \emph{The Maz'ya Anniversary Collection}, Vol.~2,
Operator Theory: Advances and Applications, Vol.~110,
Birkh\"auser, Basel, 1999, pp.~55-67.

\bibitem{Mazya2011}
V.~Maz'ya,
\emph{Sobolev Spaces with Applications to Elliptic Partial Differential Equations},
2nd ed., Springer, Berlin, 2011.

\bibitem{Henrot2006}
A.~Henrot,
\emph{Extremum Problems for Eigenvalues of Elliptic Operators},
Birkh\"auser, Basel, 2006.

\bibitem{SokolowskiZolesio1992}
J.~Soko\l{}owski and J.-P.~Zolesio,
\emph{Introduction to Shape Optimization: Shape Sensitivity Analysis},
Springer Series in Computational Mathematics, Vol.~16,
Springer, Berlin, 1992.

\end{thebibliography}
\end{document}